\documentclass[american,aps,pra,reprint,floatfix,nofootinbib,superscriptaddress,notitlepage]{revtex4-1}
\usepackage[unicode=true,pdfusetitle, bookmarks=true,bookmarksnumbered=false,bookmarksopen=false, breaklinks=false,pdfborder={0 0 0},backref=false,colorlinks=false]{hyperref}
\hypersetup{colorlinks,linkcolor=myurlcolor,citecolor=myurlcolor,urlcolor=myurlcolor}
\usepackage{graphics,epstopdf,graphicx,amsthm,amsmath,amssymb,mathptmx,braket,colortbl,color,bm,framed,mathrsfs}
\usepackage[T1]{fontenc}
\usepackage[up]{subfigure}
\usepackage{tikz}
\usepackage{algorithm}
\usepackage{algorithmic}
\usepackage{enumerate}
\usepackage{multirow}

\usepackage{tcolorbox}

\definecolor{myurlcolor}{rgb}{0,0,0.9}

\newcommand{\proj}[1]{| #1\rangle\!\langle #1 |}

\newcommand{\inner}[2]{\langle #1 , #2\rangle}

\DeclareMathOperator{\trace}{Tr}
\newcommand{\Ptr}[2]{\trace_{#1}\Pa{#2}}
\newcommand{\Tr}[1]{\Ptr{}{#1}}

\newcommand{\Pa}[1]{\left[#1\right]}

\theoremstyle{plain}
\newtheorem{thm}{Theorem}
\newtheorem{lem}[thm]{Lemma}
\newtheorem{prop}[thm]{Proposition}

\newtheorem{Def}[thm]{Definition}
\newtheorem{Rem}[thm]{Remark}

\newtheorem{Exam}[thm]{Example}

\newcommand*{\myproofname}{Proof}

\def\ot{\otimes}
\def\complex{\mathbb{C}}

\usepackage{ulem,xpatch}
\xpatchcmd{\sout}
  {\bgroup}
  {\bgroup}
  {}{}
\newcommand{\be}{\begin{equation}}
\newcommand{\ee}{\end{equation}}
\newcommand{\beq}{\begin{eqnarray}}
\newcommand{\eeq}{\end{eqnarray}}

\DeclareMathAlphabet{\mathcal}{OMS}{cmsy}{m}{n}

\makeatother

\begin{document}

\title{Magic Resource Can Enhance the Quantum Capacity of  Channels}

   \author{Kaifeng Bu}
  \email{bu.115@osu.edu}
   \affiliation{Department of Mathematics, The Ohio State University, Columbus, Ohio 43210, USA}
 \affiliation{Department of Physics, Harvard University, Cambridge, Massachusetts 02138, USA}

   \author{Arthur Jaffe}
  \email{Arthur\_Jaffe@harvard.edu}
 \affiliation{Department of Physics, Harvard University, Cambridge, Massachusetts 02138, USA}
 \affiliation{Department of Mathematics, Harvard University, Cambridge, Massachusetts 02138, USA}

\begin{abstract}
We investigate the role of magic resource in the quantum capacity of channels.
We consider the quantum channel of the recently proposed discrete beam splitter with the fixed environmental state. 
We find that if the fixed environmental state is a stabilizer state, then the quantum capacity is zero. Moreover, we 
find that the quantum capacity is nonzero for some magic states, and the quantum capacity increases linearly with 
respect to the number of single-qudit magic states in the environment. We also bound the maximal quantum capacity of the discrete beam splitter in terms of the amount of magic resource in the environmental states. These results suggest that magic resource can increase the quantum capacity of channels; it sheds new insight into the role of stabilizer  and magic states in quantum communication.
\end{abstract}

\maketitle

\section{Introduction}
Stabilizer states and circuits are basic concepts in 
discrete-variable (DV) quantum systems. They have applications ranging from use in quantum error correction codes, to understanding the possibility of a quantum computational advantage. 
 The importance of stabilizer states was recognized by Gottesman~\cite{Gottesman97} in his study of quantum error correction codes. Quantum error correction codes based on stabilizer states are 
 called stabilizer codes. Shor's  9-qubit-code~\cite{ShorPRA95} and Kitaev's toric code~\cite{Kitaev_toric} are two well-known examples of stabilizer codes.

A stabilizer vector is a common eigenstate of an abelian subgroup of the qubit Pauli group; such a vector defines a pure stabilizer state. Stabilizer circuits comprise Clifford unitaries acting on stabilizer inputs and measurements.  These circuits can be efficiently simulated on a classical computer, a result 
known as the Gottesman-Knill theorem~\cite{gottesman1998heisenberg}. Hence, non-stabilizer resources are necessary
to achieve a quantum computational advantage.

The property of not being a stabilizer has recently been called ``magic''~\cite{BravyiPRA05}. To quantify the amount of magic resource, several  measures have been proposed~\cite{Veitch12mag,Veitch14,BravyiPRL16,BravyiPRX16,HowardPRL17,bravyi2019simulation, Bu19,SeddonPRXQ21,Chen22,Bucomplexity22,BuPRA19_stat,WangNJP19,BuQST2023,BeverlandQST20,bu2022classical,LeonePRL22,WangPRA23}. These measures have been applied to the classical simulation of quantum ircuits~\cite{BravyiPRL16,BravyiPRX16,HowardPRL17,bravyi2019simulation, Bu19,SeddonPRXQ21}, to unitary  synthesis \cite{HowardPRL17,BeverlandQST20}, and to the generalization of capacity in quantum machine learning~\cite{BuPRA19_stat,BuQST2023}. 

One important measure proposed by Bravyi, Smith, and Smolin  is called  stabilizer rank~\cite{BravyiPRX16}. They used this measure to investigate time complexity in  classical simulation of quantum circuits; here  the simulation  algorithm  is based on a low-rank decomposition of the tensor products of magic states into stabilizer states.  %~\cite{BravyiPRX16}.

Pauli rank is defined as the number of nonzero coefficients in the decomposition in the Pauli basis.  This provides a lower bound on the stabilizer rank~\cite{Bu19}.  To achieve a quantum advantage for  DV quantum systems, several sampling tasks have been proposed~\cite{jozsa2014classical, koh2015further,bouland2018complexity,boixo2018characterizing,bouland2018onthecomplexity,bremner2010classical,Yoganathan19}.
Some of these proposals have been realized in experiments, which were used to claim a computational advantage over  classical supercomputers~\cite{Google19,zucongzi1,Pansci22}.

In earlier work, we introduced a discrete beam splitter unitary, enabling us to define a quantum convolution on DV quantum systems. We developed this into a framework to study DV quantum systems, including the discovery of a quantum central limit theorem that converges to a stabilizer state. 
This means that repeated quantum convolution with a given state  converges to a stabilizer state. In other words, stabilizer states can be identified as ``quantum-Gaussian'' states~\cite{BGJ23a,BGJ23b,BGJ23c,BJW24a,bu2024extremality,BGJ24a}. 
Here, we explore the channel capacity of the discrete beam splitter and investigate the role of magic resource in the channel capacity. 

% One key insight emerging from this point of view is that 
% stabilizer states are the only extremizers
% of some information theoretic inequalities~\cite{BGJ23a,BGJ23b}.
%  Furthermore, we have proposed a convolution-swap test, 
%  to determine whether a given state $\rho$ lies in the set of stabilizer states, or whether it is close to being a stabilizer state~\cite{BGJ23c}. 

In this work, we focus on the quantum capacity of a channel, which quantifies the maximal number of qubits, on average, that can be reliably transmitted.
In continuous-variable (CV) quantum  systems,  the quantum capacity of the 
beam splitter plays an important role in quantum communication. For example, in  optical communication one refers to the CV beam splitter with a thermal environment as a ``thermal attenuator channel.'' This can be generalized  to a general
 attenuator channel by choosing
non-thermal or non-Gaussian environmental states. 

There has been a surge of interest in exploring the quantum capacity of such channels~\cite{HolevoPRA01,caruso2006one,WolfPRL07,WildePRA12,pirandola2017fundamental,SabapathyPRA17,WildeIEEE18,Rosati2018narrow,JiangIEEE19,Jiang2020enhanced,wang2024passive,LamiPRL20,OskoueiIEEE22}.  This interest can be traced to the requirement of  applications in quantum information and computation, including
universal quantum computation~\cite{MenicucciPRL06,Ohliger10}, quantum error correction codes~\cite{CerfPRL09}, entanglement manipulation~\cite{Eisert02,FiurasekPRL02,GiedkePRA02,HoelscherPRA11,NamikiPRA14}, and  non-Gaussian resource theory~\cite{LamiPRA18,TakagiPRA18,LamiPRA20}.
Furthermore,  bosonic error-correcting codes, such as the Gottesman-Kitaev-Preskill code~\cite{GKPPRA01}, has been demonstrated to achieve  quantum capacity in these models up to a constant gap~\cite{JiangIEEE19}.

 A nice formula for quantum capacity based on regularized, coherent information has been obtained in the works of 
Lloyd~\cite{LloydPRA97}, Shor~\cite{shor2002quantum}, and Devetak~\cite{Devetak05}.
One surprising property of the quantum capacity is super-additivity~\cite{DiVincenzoPRA98,SmithScience08,Smith2011quantum,Cubitt2015unbounded}. This means that the quantum capacity of the tensor product of channels is larger that the sum of their individual quantum capacities, which implies
that entanglement can enhance the quantum capacity. 

In this work, we investigate the quantum capacity of DV beam splitter, with different choices of environmental states.
The results in this work reveal the role of magic states in the quantum capacity. We make a summary of the 
main results as follows:
\begin{enumerate}
\item{}
%(1) 
If the fixed environmental state is a convex combination of stabilizer states, then the
quantum capacity is zero for the discrete beam splitter with nontrivial parameters. This differs from the CV case,
where the presence of Gaussian states leads to a nonzero quantum capacity for some CV beam splitter with nontrivial parameters \cite{LamiPRL20}.

\item{}
%(2) 
We find that the quantum capacity is nonzero for some magic states, and the quantum capacity increases linearly with 
respect to the number of single-qudit magic states in the environment.  We also show that the maximal quantum capacity of the discrete beam splitter is bounded by the amount of magic resource in the environmental states. These results suggest that, in general, magic resource can increase the quantum capacity of channels as well as provide bounds on the maximal quantum capacity.

\item{}
%(3) 
We  show that environmental states, which are symmetric under the discrete phase space inversion operator, could also lead to zero quantum capacity. We  provide some magic state which satisfies this symmetry. This  shows that magic resource is necessary, but not sufficient, to increasing the quantum capacity in this model.
\end{enumerate}

\section{Basic Framework}
We focus on an $n$-qudit system with Hilbert space $\mathcal{H}^{\ot n}$. Here $\mathcal{H} \simeq \complex^d$ is  $d$-dimensional, and $d$ is a natural number. 
Let $D(\mathcal{H}^{\ot n})$ denote the set of  all quantum states on $\mathcal{H}^{\ot n}$.
We consider the orthonormal, computational basis in $\mathcal{H}$ denoted by $\set{\ket{k}}$, for ${k\in \mathbb{Z}_d}$. The Pauli $X$ and $Z$ operators
are 
\[ X: |k\rangle\mapsto |k+1 \rangle 
%\hskip -3pt{\mod d}
\;, \;\;\; Z: |k\rangle \mapsto\omega^k_d\,|k\rangle,\;\;\;\forall k\in \mathbb{Z}_d\;. 
\]
Here  $\mathbb{Z}_{d}$ is the cyclic group over $d$,  and $\omega_d=\exp(2\pi i /d)$ is a $d$-th root of unity. 
In order to define our quantum convolution, we assume $d$ is prime.

If  $d$ is an odd prime number,  the local Weyl operators (or generalized Pauli operators)
are defined as 
$
w(p,q)=\omega^{-2^{-1}pq}_d\, Z^pX^q\;.
$
Here $2^{-1}$ denotes the inverse $\frac{d+1}{2}$ of 2 in $\mathbb{Z}_d$.  
In the $n$-qudit system, the Weyl operators are defined as
$
w(\vec p, \vec q)
=w(p_1, q_1)\ot...\ot w(p_n, q_n),
$
with $\vec p=(p_1, p_2,..., p_n)\in \mathbb{Z}^n_d$, and $\vec q=(q_1,..., q_n)\in \mathbb{Z}^n_d $.  

Denote $V^n:=\mathbb{Z}^n_d\times \mathbb{Z}^n_d$; this represents the phase space for $n$-qudit systems, in analogy with continuum mechanics~\cite{Gross06}. The functions $w(\vec p, \vec q)$ on phase space form an orthonormal basis 
 with respect to the inner product 
$\inner{A}{B}=\frac{1}{d^n}\Tr{A^\dag B}$. 

The characteristic function $\Xi_{\rho}:V^{n}\to\complex$ of a quantum state $\rho$ is
\begin{eqnarray*}
\Xi_{\rho}(\vec{p},\vec q):=\Tr{\rho w(-\vec{p},-\vec q)}.
\end{eqnarray*}
Hence, 
any quantum state $\rho$ can be written as a linear combination of the Weyl operators 
$
\rho=\frac{1}{d^n}
\sum_{(\vec{p},\vec q)\in V^n}
\Xi_{\rho}(\vec{p},\vec q)w(\vec{p},\vec q)\;.
$
The transformation from the computational basis to the Pauli basis is the quantum Fourier transform that we consider. It has found extensive uses in a myriad of applications, including discrete Hudson theorem~\cite{Gross06}, quantum Boolean functions~\cite{montanaro2010quantum}, quantum circuit complexity~\cite{Bucomplexity22}, quantum scrambling~\cite{GBJPNAS23}, the generalization capacity 
 of quantum machine learning~\cite{BuPRA19_stat}, and
  quantum state 
 tomography~\cite{Bunpj22}.

Stabilizer states are an important family of quantum states; such a state is invariant under an abelian subgroup of the Pauli group. Specifically, a pure stabilizer vector $\ket{\psi}$ for $n$ qubits is the common eigenvector of a commuting  subgroup  with  $n$ generators, $ \set{g_i}_{i\in [n]}$, so that $g_i\ket{\psi}=\ket{\psi}$ for each $i$. 
The corresponding density matrix can be expressed 
as $\proj{\psi}=
\Pi^n_{i=1}\mathbb{E}_{k_i\in \mathbb{Z}_d}\, g^{k_i}_i$, where the expectation is defined as  $\mathbb{E}_{k_i\in \mathbb{Z}_d}\,g^{k_i}_i:=\frac{1}{d}\sum_{k_i\in \mathbb{Z}_d}\,g^{k_i}_i$.

A mixed state $\rho$ is a stabilizer state if there exists some abelian  subgroup of Pauli operators  with $r< n$ generators $ \set{g_i}_{i\in [r]}$ such that $\rho=\frac{1}{d^{n-r}}
\Pi^r_{i=1}\mathbb{E}_{k_i\in \mathbb{Z}_d}\, g^{k_i}_i$. We take STAB to denote the set of all  stabilizer states, which is also called the set of minimal stabilizer-projection states.  The set $STAB $ is the set of states which
are a convex combination of pure stabilizer states. One magic measure, which quantifies the amount of magic resource in quantum states, is called 
the  relative entropy of magic (See Definition 35 in~\cite{BGJ23c}), denoted as $MRM(\rho)$. This  is defined as 
\begin{align}\label{eq:MRM}
MRM(\rho):=\min_{\sigma\in STAB} D(\rho||\sigma)\;.
\end{align}
Here $D(\rho||\sigma)=\Tr{\rho\log\rho}-\Tr{\rho\log \sigma}$ is the quantum relative entropy. We also consider $MRM_{\infty}(\rho):=\min_{\sigma\in STAB} D_{\infty}(\rho||\sigma)$ with maximal relative entropy
$D_{\infty}(\rho||\sigma)=\min\set{\lambda:\rho\leq 2^{\lambda}\sigma}$~\cite{BGJ23c}.

Denote the vector $| \vec i\, \rangle = |  i_1 \rangle \otimes \cdots \otimes |  i_n \rangle \in \mathcal{H}^{\otimes n} $.  In order to define the discrete beam splitter, consider the tensor product Hilbert space of two $n$-qudit Hilbert spaces $\mathcal{H}_A\ot\mathcal{H}_B $.

{\bf Discrete Beam Splitter~\cite{BGJ23a,BGJ23b}:}
Given a prime $d$ and $s,t\in \mathbb{Z}_d$ satisfying  $s^2+t^2\equiv 1 \mod d$, the discrete beam splitter unitary $U_{s,t}$ for a $2n$-qudit system $\mathcal{H}_A\ot\mathcal{H}_B $ is
\begin{align}\label{1231shi1}
 U_{s,t} = \sum_{\vec i,\vec j\in \mathbb{Z}^n_d} |s\vec i+t\vec j \rangle \langle \vec i|_A \otimes |- t\vec i+s\vec j\rangle \langle \vec j|_B\;.
 \end{align}
The quantum channel $  \Lambda_{s,\sigma}$ with a fixed environmental state $\sigma$ is,
\begin{eqnarray}\label{eq:chan_DBS}
    \Lambda_{s,\sigma}(\rho):= \Ptr{B}{ U_{s,t} (\rho \otimes \sigma) U^\dag_{s,t}}\;.
\end{eqnarray}
The complementary channel is 
\begin{eqnarray}
    \Lambda^c_{s,\sigma}(\rho)=\Ptr{A}{ U_{s,t} (\rho \otimes \sigma) U^\dag_{s,t}}.
\end{eqnarray}

We denote $\frac{s^2 \mod d}{d}$ to be the discrete transmission rate, and $\frac{t^2 \mod d}{d}$ to be the discrete reflection rate. We summarize the properties of the discrete beam splitter  in Appendix~A.
In this work, we focus on the discrete beam splitter with nontrivial parameters, i.e., $s^2,t^2\not\equiv 0,1 \mod d$.

One motivation to consider the discrete beam splitter is to study the quantum additive noise model of qudit systems.  
One important model for noise in classical communication is additive noise. Additive noise is defined as $Y=X+Z$, where
$X, Z$ are independent random variables, and the probability distribution
of the output $Y$ is the classical convolution of $X$ and $Z$. In other words,
$P_Y(a)=\sum_xP_X(x)P_Z(a-x)$. Hence, it is natural to ask the question: what is the quantum additive noise 
channel on qudit or qubit systems? This is the motivation for us to consider the
 discrete beam splitter. We believe that this represents a good candidate for a quantum additive noise channel on a qudit system. One reason is that the discrete beam splitter provides a good quantum convolution, and it will reduce to the classical additive nosie channel  when the input states are diagonal  in the computational basis, i.e., the classical states.

\begin{Rem}
   In a qubit-system with $d=2$, there is no nontrivial choice of $s,t$ such that $s^2+t^2\equiv 1\mod 2$, since in that case  
$(s,t)$ could only be $(0,1)$ or $(1,0)$. Hence, it is impossible to consider the discrete beam splitter 
with two input states. We give an alternative in~\cite{BGJ23c}.
\end{Rem}

\section{Main results}
The quantum capacity of 
a channel $\Lambda$ can be written as a regularized form of the coherent information, as given by the Lloyd-Shor-Devetak theorem~\cite{LloydPRA97,shor2002quantum,Devetak05}:
\begin{eqnarray}
    Q(\Lambda)&:=&\lim_{N\to \infty}\frac{Q^{(1)}(\Lambda^{\ot N})}{N}\;,\\
    Q^{(1)}(\Lambda)&:=&\max_{\rho\in D(\mathcal{H}^{\ot n})}I_c(\rho,\Lambda)\;.
\end{eqnarray}
Here the coherent information is 
$I_c(\rho,\Lambda):=S(\Lambda(\rho))-S(I\ot\Lambda(\proj{\Psi}_{RA}))$ with the  purification$\Psi_{RA}$ of $\rho$ \footnote{This means $\Ptr{R}{\proj{\Psi_{RA}}}=\rho$.} .
It can also be written as  
$I_c(\rho,\Lambda)=S(\Lambda(\rho))-S(\Lambda^c(\rho))$, where $\Lambda^c$ is the complementary channel of $\Lambda$. In general, the optimization over all states in the asymptotic regime makes it 
difficult to calculate the quantum capacity.

Let us now consider the quantum capacity of the quantum channel $\Lambda_{s,\sigma}$ defined in \eqref{eq:chan_DBS} using  the discrete beam splitter
and the fixed environmental state $\sigma$. We start by considering $\sigma$ to be a convex combination of stabilizer states, in order to explore the role of stabilizers.

\begin{thm}[\bf Stabilizer environments yield zero quantum capacity of discrete beam splitters] \label{thm:main_1}
Let nontrivial $s,t\in \mathbb{Z}_d$ satisfy $s^2+t^2\equiv 1\mod d$, and  the environmental state
 $\sigma$  be a convex combination of stabilizer states. Then  
\begin{eqnarray}
    Q(\Lambda_{s,\sigma})=0.
\end{eqnarray}
\end{thm}
We prove Theorem \ref{thm:main_1} in Appendix~B. This shows that stabilizer environmental states ensure that  the quantum capacity becomes zero for any nontrivial parameters $s,t$. This phenomenon differs from the CV case in the following way:  the CV beam splitter with a pure Gaussian state (e.g., the vacuum state) and with transmissivity $\lambda>1/2$ has a quantum capacity  strictly larger than zero~\cite{LamiPRL20}.

From Theorem \ref{thm:main_1}, we infer that a magic environmental state is necessary in order to obtain nonzero quantum capacity. Hence, let us consider the case in which the environmental state is a magic state. Let $\sigma_k$ be  a quantum state generated by a Clifford circuit $U_{cl}$ on $k$ copies of  1-qudit magic state $\ket{\text{magic}}$, namely 
\begin{eqnarray}\label{eq:state_form}
    \sigma_k=U_{\text{cl}}(\proj{\text{magic}}^{\ot k}\ot\proj{0}^{n-k})U^\dag_{\text{cl}}\;.
\end{eqnarray}

\begin{thm}[\bf Magic resource can enhance quantum capacity of discrete beam splitters]\label{thm:main2}
Given  nontrivial $s,t\in \mathbb{Z}_d$ with $s^2+t^2\equiv 1\mod d$, there exists 
some 1-qudit magic state $\ket{\text{magic}}$ and 
a universal constant $c>0$, independent of $d$ and $n$, 
such that 
\begin{eqnarray}
Q(\Lambda_{s,\sigma_k})\geq kc\;.
%\quad \text{for } c>0\;.
\end{eqnarray}
Here, the state $\sigma_k$ is given by \eqref{eq:state_form}.
\end{thm}

We prove Theorem \ref{thm:main2}  in Appendix~C; here we sketch the idea for a very special case. 
Choose the  1-qudit magic state for the discrete beam splitter to be 
\begin{eqnarray}\label{eq:magic_stat}
    \ket{\sigma}_B=\frac{1}{\sqrt{2}}(\ket{0}_B+\ket{1}_B)\;,
\end{eqnarray}
and let  the input state  be
$\rho_A=\frac{1}{2}(\proj{0}_A+\proj{t^{-1}s}_A)$.  Then
$$\Lambda_{s,\sigma}(\rho_A)=\frac{1}{4}(\proj{0}_A+\proj{t}_A+\proj{t^{-1}s^2}_A+\proj{t^{-1}}_A),$$ and 
$$\Lambda^c_{s,\sigma}(\rho_A)=\frac{1}{2}\proj{0}_B+\frac{1}{4}\proj{-s}_B+\frac{1}{4}\proj{s}_B\;,$$
for $s^2\not\equiv t^2\mod d$. Hence,  the coherent information can be written as the entropy difference 
$$I_c(\rho_A,\Lambda_{s,\sigma})=S(\Lambda_{s,\sigma}(\rho_A))-S(\Lambda^c_{s,\sigma}(\rho_A))=\frac{1}{2}\;.$$
Hence, the quantum capacity $Q(\Lambda_{s,\sigma})\geq \frac{1}{2}$, and $Q(\Lambda_{s,\sigma^{\ot k}\ot \proj{0}^{n-k}})\geq kQ(\Lambda_{s,\sigma})\geq \frac{k}{2}$.
If $s^2\equiv t^2\mod d$, we need some additional arguments.

In addition, we find that  the maximal quantum capacity of $\Lambda_{\sigma}$
 by the amount of  magic of $\sigma$.

\begin{thm}[\bf Magic bound on quantum capacity for discrete beam splitters]\label{thm:main_bound}
    Given  nontrivial $s,t\in \mathbb{Z}_d$ with $s^2+t^2\equiv 1\mod d$, we have
    \begin{align}\label{eq:main_bound}
        Q(\Lambda_{\sigma})
        \leq MRM(\sigma),
    \end{align}
    where $ MRM(\sigma)$ is a magic measure defined in \eqref{eq:MRM}
\end{thm}
We prove Theorem \ref{thm:main_bound} in  Appendix C of the supplementary material.  This bound indicates that if one wishes to achieve higher quantum capacity,  the magic resource of the environmental state $\sigma$ should be large enough. This result puts a fundamental limit on the maximal 
quantum capacity. One interesting question is to find what  optimal magic states achieve the maximal quantum capacity. By Theorem~\ref{thm:main_bound}, 
the states  that are the  extremalizers of the inequality \eqref{eq:main_bound} are the optimal magic states. Hence, identifying these optimal magic states involves understanding the states that maximize the bound established in the theorem.

 Here, let us consider an example that 
achieves the upper bound up to a constant factor. 
Let us now consider environmental states $\sigma_k$ of the form \eqref{eq:state_form}, along with a  
 1-qudit magic state of the form \eqref{eq:magic_stat}. Then the magic measure of $\sigma_k$ is
 $MRM(\sigma_k)=k\log d$,where $d$ is the local dimension of the qudit; this is a fixed constant.  Hence,  by Theorems~\ref{thm:main2} and ~\ref{thm:main_bound}, 
$Q(\Lambda_{\sigma_k})=\Theta(MRM(\sigma_k))$, where
 $f=\Theta(g)$ means there exist constants $ c_1, c_2$ such that $c_1g\leq f\leq c_2g$. This provides an example 
 that achieves the maximal quantum capacity satisfying the equality in \eqref{eq:main_bound}, 
 up to some constant factor.

Let us consider the discrete phase space point operators and their corresponding symmetries. The discrete phase space point operator $A(\vec{p}, \vec q)$ with $(\vec{p}, \vec q)\in V^n$ is defined as $A(\vec p, \vec q)=w(\vec p, \vec q)A w(\vec p, \vec q)^\dag$, where  $A = \frac{1}{d^n}\sum_{(\vec{u}, \vec v)\in V^n}w(\vec{u}, \vec v)$.  These operators can be used to define the discrete Wigner function, where the nongativity of the discrete Wigner function is used to characterize the stabilizer states on qudit systems~\cite{Gross06}.
The operator $A$ can be rewritten as $A=\sum_{\vec{x}}\ket{-\vec{x}}\bra{\vec{x}}$.
In Appendix~A, we list some properties of these discrete phase-space point operators for completeness.

Now let us consider  symmetry under the discrete phase space point operators.
A quantum state $\sigma$ is defined to be symmetric under the discrete phase space inverse operation if
\begin{eqnarray}\label{eq:sym_ps}
    A\sigma A^\dag=\sigma.
\end{eqnarray}
For example, the zero-mean stabilizer states\footnote{A zero-mean stabilizer state is a stabilizer state with the characteristic function taking values either 0 or 1~\cite{BGJ23a,BGJ23b}.} exhibit this symmetry as the characteristic function of the zero-mean stabilizer states is either 0 or 1.

 We now demonstrate that choosing  symmetric states as environmental states will lead to zero quantum capacity, even though these states could be magic states.  

\begin{thm}[\bf Symmetry can  limit 
 quantum capacity of balanced beam splitters]\label{thm:main3}
     Let $\sigma$ be an $n$-qudit state, which is a convex combination of states $w(\vec a)\sigma_{\vec a} \,w(\vec a)^\dag$, with each state $\sigma_{\vec a}$ having discrete, phase-space, inverse symmetry.
    Then $\Lambda_{s,\sigma}$ is anti-degradable\footnote{A channel $\Lambda$ is called anti-degradable if there exists a CPTP map $\Gamma$ such that 
    $\Lambda=\Gamma\circ\Lambda^c$. Similarly, a channel $\Lambda$ is called degradable if there exists a CPTP map $\Gamma$ such that 
    $\Lambda^c=\Gamma\circ\Lambda$.
} for $s\equiv t\mod d$,  and the quantum capacity
    \begin{eqnarray}
        Q(\Lambda_{s,\sigma})=0\;.
    \end{eqnarray}
\end{thm}
The proof of Theorem  \ref{thm:main3} is presented in Appendix~D. Combined with the Theorem  \ref{thm:main_1}, \ref{thm:main2}, and \ref{thm:main3},
we conclude that magic resource is necessary but not sufficient to increase the quantum capacity of the quantum channel defined by the
discrete beam splitter.  

Note that, 
several results show the extremality of stabilizer states in channel capacity,
such as the classical capacity. 
The Holevo capacity is an important quantity that provides a least upper bound on the
classical capacity; this is known as the Holevo-Schumacher-Westmoreland theorem~\cite{HolevoIEEE98,SchumacherPRA97}. 
It is known that  stabilizer states are the only extremizers of the Holevo capacity given by the discrete beam splitter, that is,
the quantum channel $\Lambda_{\sigma}$ in \eqref{eq:chan_DBS} achieves its maximal Holevo capacity  $\sigma$, if and only if  $\sigma$ is a pure stabilizer state. (See Theorem 19 in \cite{BGJ23a} and Theorem 73 in \cite{BGJ23b} for the details.) 

However, this is not the case for quantum capacity as shown in Theorem~\ref{thm:main3}. Here, we  give  an example of a magic state with zero quantum capacity:

\begin{Exam}
Consider the  1-qudit state 
\begin{eqnarray}
    \ket{\sigma}_B=\frac{1}{\sqrt{2}}(\ket{1 \mod d}_B+\ket{-1 \mod d}_B)\;.
\end{eqnarray}
Since the local dimension $d$ is an odd prime number, 
$\ket{\sigma}_B$ is a magic state. This state is also symmetric under the phase inverse operation in \eqref{eq:sym_ps}.
Hence from Theorem~\ref{thm:main3} we infer that the quantum capacity of $\Lambda_{\sigma}$ is zero for  $s\equiv t \mod d$.

\end{Exam}

It is  interesting to  identify all the  magic states which can play a beneficial
role in increasing the quantum capacity. 
 Theorem~\ref{thm:main3}
holds for the 
balanced discrete beam splitter, but  may not hold for other cases. So, identifying  magic states that increase the quantum capacity depends on the 
parameters of the discrete beam splitter. Moreover,
based on Theorem \ref{thm:main_bound},   one candidate is the family of magic states saturating the equality  in \eqref{eq:main_bound} (up to some constant factor), i.e., 
the magic state $\sigma$ with $Q(\Lambda_{\sigma})=\Theta(MRM(\sigma))$.
This can be achieved by the
 example we discuss after Theorem \ref{thm:main_bound}, which may lead to 
 finding other examples.

Finally, we briefly discuss the 
    connection between our work and quantum error correction code by
reinterpreting the results in terms of quantum coding. 
The detailed derivation of the  following results are
 presented in Appendix E. 
We explore how entanglement fidelity varies with different environmental states $\sigma$; we construct some encoding to show there exist magic states 
that will increase the entanglement fidelity. We also provide an upper bound on the advantage of entanglement fidelity in terms of the amount of magic resource.

 Consider the encoding  $\mathcal{E}_K:\mathcal{H}_S\to \mathcal{H}_A$ and the decoding $\mathcal{D}_K:\mathcal{H}_A\to \mathcal{H}_S$, where $\mathcal{H}_S=\complex^K$ is the logical space with dimension $K$ and  $\mathcal{H}_A=(\complex^d)^{\ot n}$ is the physical space.   
Entanglement fidelity, crucial in quantifying the performance of the 
quantum error correction code,  is given for discrete beam splitter $\Lambda_{\sigma}$ by
$F_e(\mathcal{E}_K, \mathcal{D}_K, \sigma)
    =\bra{\Phi}\mathcal{D}_K\circ\Lambda_{\sigma}\circ\mathcal{E}_K(\proj{\Phi})\ket{\Phi}
$, where $\ket{\Phi}=\frac{1}{\sqrt{K}}\sum_{i\in K}\ket{i}_R\ket{i}_S$ being maximally entangled state on $\mathcal{H}_R\ot\mathcal{H}_S$.

We first consider the maximal entanglement fidelity over all stabilizer states, i.e., $\max_{\tau\in STAB}\max_{\mathcal{E}_K} \max_{\mathcal{D}_K}   F_e(\mathcal{E}_K, \mathcal{D}_K, \tau)$, which quantifies the optimal performance of entanglement 
fidelity over all stabilizer states. We find that
\begin{align}
    \max_{\tau\in STAB}\max_{\mathcal{E}_K} \max_{\mathcal{D}_K}   F_e(\mathcal{E}_K, \mathcal{D}_K, \tau)
    =\frac{1}{K}.
\end{align}
Hence this quantity is usually very small $\leq 1/2$, as the logical dimension $K\geq 2$.

Moreover, let us consider the environmental state $\ket{\sigma}_B=\frac{1}{\sqrt{2}}(\ket{0}_B+\ket{t}_B)$, which is magic in qudit-system. 
For $K=2$, i.e., $\mathcal{H}_S$ is a logical qubit, let us consider the encoding
$\mathcal{E}_2$ as follows
\begin{align}
      \mathcal{E}_2:  \ket{0}_S\to \ket{0}_A, \quad
    \ket{1}_S\to \ket{s}_A.
\end{align}
where $\set{\ket{0}_S,\ket{1}_S}$ is an orthonormal basis of $\mathcal{H}_S$.
 Then there exists a decoding $\mathcal{D}_2$ such that 
$F_e(\mathcal{E}_2, \mathcal{D}_2, \sigma_B)=\frac{3}{4}>\frac{1}{2}=\max_{\tau\in STAB}\max_{\mathcal{E}_2} \max_{\mathcal{D}_2}   F_e(\mathcal{E}_2, \mathcal{D}_2, \tau)$.

Furthermore,  we find that the advantage of  magic states on the performance of 
entanglement fidelity compared to stabilizer states is bounded by the amount of magic resource as follows
\begin{align}
    \frac{\max_{\mathcal{E}_K} \max_{\mathcal{D}_K}   F_e(\mathcal{E}_K, \mathcal{D}_K, \sigma)}{\max_{\tau\in STAB}\max_{\mathcal{E}_K} \max_{\mathcal{D}_K}   F_e(\mathcal{E}_K, \mathcal{D}_K, \tau)}
    \leq 2^{MRM_{\infty}(\sigma)}.
\end{align}
This is based on 
 the convexity of $F_e$ and also the definition of magic measure 
$MRM_{\infty}(\rho)=\min_{\sigma\in STAB}D_{\infty}(\rho||\sigma)$.

\section{Conclusion and future work}
In this letter, we provide new understanding of stabilizer and magic states in quantum communication. We show  that magic resource can increase the quantum channel
capacity in the model defined by the discrete beam splitter, and the maximal quantum capacity is bounded by the amount of magic resource in the environmental state.

One intriguing problem for further study is the full characterization of magic states saturating the equality 
in \eqref{eq:main_bound}, up to some constant factor. This will help us to have a better understanding on the structure of 
magic states that can achieve maximal quantum capacity and increase quantum capacity.

Moreover, it is natural for future work to consider
different  quantum channel capacities of the discrete 
beam splitter, such as private capacity~\cite{LiPRL09,Junge18}.  Besides,
in qubit systems,  it is impossible to consider the discrete beam splitter for  two input states and with nontrivial parameters $s,t$; so it would  be interesting 
to consider the channel capacities of  $n$-qubit channels to show the power of
magic states.

In addition, we can also generalize the results in this work to non-Clifford operations. In that case we need consider  magic resource, not only in the environmental state, but also in the non-Clifford operation. 
The quantum capacity may not only depend on the magic resource of the environmental state, but also the magic resource of the non-Clifford operation. This will be an interesting problem to investigate in a future work.

Furthermore, we have briefly discussed the connection between the results and its connection 
with quantum error correction code. 
To find further potential application in quantum communication is one intriguing direction. The study of the Gottesman-Kitaev-Preskill code in the  CV beam splitter~\cite{JiangIEEE19,wang2024passive}
may be helpful for investigating this question. We leave this questions for further exploration.

\section{Acknowledgements}
We thank Liyuan Chen, Roy Garcia, Weichen Gu, Liang Jiang, Marius Junge, Bikun Li, Zhaoyou Wang, Zixia Wei, and Chen Zhao  for discussions.
This work was supported in part by the ARO Grant W911NF-19-1-0302 and the ARO
MURI Grant W911NF-20-1-0082.

\bibliography{reference}{}

\clearpage
\newpage
\onecolumngrid
\begin{appendix}
\section{Properies of the discrete Beam splitter}\label{appen:prop_dbs}
Unless noted otherwise, here we summarize some results in~\cite{BGJ23a,BGJ23b}. For simplicity denote  $\Xi_{\rho}(\vec x)$ with $\vec x=(\vec p, \vec q)\in V^n$.

\begin{prop}[Proposition 35 in \cite{BGJ23b}]
    For any  $\vec a, \vec b\in V^n$, the discrete beam splitter $U_{s,t}$ satisfies 
    \begin{eqnarray}
        U_{s,t}(w(\vec a)\ot w(\vec b))U^\dag_{s,t}
        =w(s\vec a+t\vec b)\ot w(-t\vec a+s\vec b)\;.
    \end{eqnarray}
    \end{prop}

\begin{Def}[\bf Quantum convolution defined by discrete beam splitter]
    The quantum convolution  of two $n$-qudit states $\rho$ and $\sigma$ is 
\begin{align}\label{eq:conv_B}
\rho \boxtimes_{s,t} \sigma =\Lambda_{\sigma}(\rho)= \Ptr{B}{ U_{s,t} (\rho \otimes \sigma) U^\dag_{s,t}}.
\end{align}
And the complementary one 
\begin{align}\label{eq:conv_C}
\rho \tilde{\boxtimes}_{s,t} \sigma =\Lambda^c_{\sigma}(\rho)= \Ptr{A}{ U_{s,t} (\rho \otimes \sigma) U^\dag_{s,t}}.
\end{align}
\end{Def}

\begin{Def}[\bf Mean state]\label{def:mean_state}
Given an $n$-qudit state $\rho$,  the mean state  $\mathcal{M}(\rho)$ is the 
operator with the characteristic function: 
\begin{align}\label{0109shi6}
\Xi_{\mathcal{M}(\rho)}(\vec x) :=
\left\{
\begin{aligned}
&\Xi_\rho ( \vec x) , && |\Xi_\rho ( \vec x)|=1,\\
& 0 , && |\Xi_\rho (  \vec x)|<1.
\end{aligned}
\right.
\end{align}
 The mean state $\mathcal{M}(\rho)$ is a stabilizer state.
\end{Def}

\begin{lem}[\cite{Gross06}]
The set of phase space point operators $\set{A(\vec{p}, \vec q)}_{(\vec{p}, \vec q)\in V^n}$
satisfies three properties when the local dimension $d$ is an odd prime:
\begin{enumerate}[(1)]
\item{}
 $\set{A(\vec{p}, \vec q)}_{(\vec{p}, \vec q)\in V^n}$ forms a Hermitian, orthonormal basis with respect to the inner product defined by 
$\inner{A}{B}=\frac{1}{d^n}\Tr{A^\dag B}$.
\item{}
$A(\vec{0},\vec 0)=\sum_{\vec{x}}\ket{-\vec{x}}\bra{\vec{x}}$ in the Pauli $Z$ basis.
\item{}
$A(\vec{p}, \vec q)=w(\vec{p}, \vec q)A(\vec{0}, \vec 0)w(\vec{p}, \vec q)^\dag$.
\end{enumerate}
\end{lem}
In the main context, we denote $A(\vec 0, \vec 0)$ as $A$ for simplicity.
\begin{lem}
    The zero-mean stabilizer state is symmetric under the phase space inverse operator.
\end{lem}
\begin{proof}
By simple calculation, we have 
\begin{eqnarray}
    AZA^{\dag}=\sum_{j}\omega^j_dA\proj{j}A^\dag=\sum_{j}\omega^j_d\proj{-j}=Z^{-1}, \\
    AXA^{\dag}=\sum_{j}A\ket{j+1}\bra{j}A=\sum_{j}\ket{-j-1}\bra{-j}=X^{-1},
\end{eqnarray}
which implies that 
\begin{eqnarray}
    Aw(\vec p, \vec q)A^\dag=w(-\vec p, -\vec q).
\end{eqnarray}
Hence, we have 
\begin{eqnarray}
    \Xi_{A\rho A^\dag}(\vec x)=\Xi_{\rho}(-\vec x), \quad\forall \vec x\in V^n.
\end{eqnarray}
Hence, the state $\rho$ is symmetric under the phase space inverse operator iff 
$\Xi_{\rho}(\vec x)=\Xi_{\rho}(-\vec x)$, for any $\vec x\in V^n$.
For a zero-mean stabilizer state $\rho$, we have 
$\Xi_{\rho}(\vec x)=\Xi_{\rho}(-\vec x)$, which is either 0 or 1. Therefore, it is symmetric under the phase space inverse operator.

\end{proof}

 The  phase space point operators  can be used to define the discrete Wigner function 
 \begin{eqnarray}
     W_{\rho}(\vec p, \vec q)=\Tr{\rho A(\vec p, \vec q)}.
 \end{eqnarray}
 One important result about the discrete Wigner function is the discrete Hudson theorem~\cite{Gross06}. It states that for any 
 $n$-qudit pure state $\psi$ with odd prime $d$, it is a stabilizer state iff the discrete Wigner function $W_{\psi}$ is nonnegative.

\begin{lem}\label{lem:key_tech}
The quantum convolution $\boxtimes_{s,t}$ satisfies the following properties:
     \begin{enumerate}
    \item {\bf Convolution-multiplication duality:} $\Xi_{\rho\boxtimes_{s,t}\sigma}(\vec x)=\Xi_{\rho}(s\vec x)\Xi_{\sigma}(t\vec x)$, for any $\vec x\in V^n$. (See Proposition 11 in \cite{BGJ23a}.)
    \item  {\bf Convolutional stability:} If both $\rho$ and $\sigma$ are stabilizer states,  then 
$\rho\boxtimes_{s,t}\sigma$ is still a stabilizer state. (See Proposition 12 in \cite{BGJ23a}.)
    \item {\bf Quantum central limit theorem:} Let 
$\boxtimes^N\rho:=(\boxtimes^{N-1}\rho)\boxtimes\rho$ be the $N^{\rm th}$ repeated quantum convolution, and  $\boxtimes^2\rho=\rho\boxtimes\rho$. Then 
$\boxtimes^N\rho$ converges to a stabilizer state $\mathcal{M}(\rho)$  as $N\to \infty$. (See Theorem 24 in~\cite{BGJ23a}.)

\item {\bf Quantum maximal entropy principle:} $S(\rho)\leq S(\mathcal{M}(\rho))$. (See Theorem 4 in~\cite{BGJ23a}.) In general, $S(\rho\boxtimes\sigma)\geq \max\set{S(\rho), S(\sigma)}$ (See Theorem 14 in~\cite{BGJ23a}.) 

\item{\bf Commutativity with Clifford unitaries: }For any Clifford unitary $U$, there exists some Clifford unitary $V$ such
that $(U\rho U^\dag)\boxtimes (U\sigma U^\dag)=V(\rho\boxtimes \sigma)V^\dag$ for any input states $\rho$ and $\sigma$. (See Lemma 85  in~\cite{BGJ23b}.) Similarly, $(U\rho U^\dag)\tilde{\boxtimes} (U\sigma U^\dag)=V(\rho\tilde{\boxtimes} \sigma)V^\dag$.

\item{\bf Wigner function positivity:} If $s\equiv t\mod d$, then the discrete Wigner function $W_{\rho\boxtimes\sigma}\geq 0$
for any n-qudit states $\rho$ and $\sigma$. (See Remark 83 in~\cite{BGJ23b}.)
\end{enumerate}
\end{lem}

\section{Discrete beam splitter with stabilizer environment state }\label{appen:stab}
\begin{lem}\label{lem:convex_Coh}
Given a quantum channel $\Lambda=\sum_ip_i\Lambda_i$ acting on the system $\mathcal{H}_A$, we have the following convexity of coherent information,
   \begin{eqnarray}
 S(\Lambda(\rho))-S(I\ot\Lambda(\proj{\Psi_{RA}}))
 \leq  \sum_ip_i[S(\Lambda_i(\rho))-S(I\ot\Lambda_i(\proj{\Psi_{RA}})],
   \end{eqnarray}
   where $\Psi_{AR}$ is a purification of $\rho$ on the system $\mathcal{H}_A$.
\end{lem}
\begin{proof}
By the joint convexity of relative entropy, we have 
   \begin{eqnarray}
    D(I\ot\Lambda(\proj{\Psi_{RA}}||\frac{I_R}{d_R}\ot \Lambda(\rho)))
\leq \sum_ip_iD(I\ot\Lambda_i(\proj{\Psi_{RA}}||\frac{I_R}{d_R}\ot \Lambda_i(\rho))),
\end{eqnarray}
where $d_R$ is the dimension of the ancilla system $\mathcal{H}_R$. This
 is equivalent to
\begin{eqnarray}
 S(\Lambda(\rho))-S(I\ot\Lambda(\proj{\Psi_{RA}}))
 \leq  \sum_ip_i[S(\Lambda_i(\rho))-S(I\ot\Lambda_i(\proj{\Psi_{RA}})].
\end{eqnarray}

\end{proof}

\begin{thm}[Restatement of Theorem 2]\label{thm:1_APP}
Given nontrivial $s,t\in \mathbb{Z}_d$ with $s^2+t^2\equiv 1\mod d$, and the fixed environment state
 $\sigma$  being any convex combination of stabilizer states, we have 
\begin{eqnarray}
    Q(\Lambda_{s,\sigma})=0.
\end{eqnarray}
\end{thm}
\begin{proof}
Let us first prove that  $Q^{(1)}(\Lambda_{s,\sigma})=0$, which means
that for any $n$-qudit state $\rho$ with the purification $\Psi_{RA}$, 
\begin{eqnarray}
    S(\Lambda(\rho))-S(I\ot\Lambda(\proj{\Psi_{RA}}))\leq 0.
\end{eqnarray}
By Lemma \ref{lem:convex_Coh}, we only need to consider the case where $\sigma$ is a pure stabilizer state on an $n$-qudit system.
Then, there exists some abelian group of Weyl operators with $n$ generators $G_{\sigma}:=\set{w(\vec x): w(\vec x)\sigma=\sigma}$.
By a simple calculation, $\Delta(\rho):=\frac{1}{d^n}\sum_{w(\vec x)\in G_{\sigma}}w(\vec x)\rho w(\vec x)^\dag$ is the full-dephasing channel, and there exists 
an orthonormal basis $\ket{\varphi_{\vec a}}$ such that 
$$w(\vec x)\ket{\varphi_{\vec a}}=\omega^{\vec x\cdot\vec a}_d\ket{\varphi_{\vec a}}, $$ for any $w(\vec x)\in G_{\sigma}$, and 
$$\Delta(\rho)=\frac{1}{d^n}\sum_{\vec a\in \mathbb{Z}^n_d}\bra{\varphi_{\vec a}}\rho\ket{\varphi_{\vec a}}\proj{\varphi_{\vec a}}.$$ Note that, $\sigma=\proj{\varphi_{\vec 0}}$. Let us denote $p_{\vec a}=\bra{\varphi_{\vec a}}\rho\ket{\varphi_{\vec a}}$.
Then, we have 
\begin{eqnarray}
\Lambda_{\sigma}(\rho)=\Lambda(\rho\ot\sigma)
=\frac{1}{d^n}\sum_{w(\vec x)\in G_{\sigma}}\Lambda(\rho\ot w(\vec x)\sigma w(\vec x)^\dag)
=\Lambda\left(\frac{1}{d^n}\sum_{w(\vec x)\in G_{\sigma}}w(\vec x)\rho w(\vec x)^\dag\ot \sigma \right)
=\Lambda_{\sigma}(\Delta(\rho)),
\end{eqnarray}
where the second equality comes from the definition of $G_{\sigma}$, and 
the 
third equality comes from the following property (proved as Proposition 41 in\cite{BGJ23b}), namely 
\begin{align*}
\Lambda(w(\vec x)\ot w(\vec y)\rho_{AB}w(\vec x)^\dag\ot w(\vec y)^\dag)
=w(s\vec x+t\vec y)\rho_{AB}w(s\vec x+t\vec y)^\dag,
\end{align*}
where $\Lambda(\rho_{AB})=\Ptr{B}{U_{s,t}\rho_{AB}U^\dag_{s,t}}$. 
Moreover, based on the convolution-multiplication duality in Lemma \ref{lem:key_tech}, 
we have 
$$\Xi_{\sigma\boxtimes_{s,t}\proj{\varphi_{\vec a}}}(\vec x)
=\Xi_{\proj{\varphi_{\vec 0}}\boxtimes_{s,t}\proj{\varphi_{\vec a}}}(\vec x)
=\Xi_{\proj{\varphi_{\vec 0}}}(s\vec x)\Xi_{\proj{\varphi_{\vec a}}}(t\vec x)
=\omega^{-t\vec a\cdot\vec x}_d\delta_{\vec x\in G_{\sigma}}
=\Xi_{\proj{\varphi_{t\vec a}}}(\vec x).
$$
Hence, 
\begin{eqnarray}
 \Lambda_{\sigma}(\proj{\varphi_{\vec a}})  = \sigma\boxtimes_{s,t}\proj{\varphi_{\vec a}}
 =\proj{\varphi_{t\vec a}}_A.
\end{eqnarray}
Thus,
\begin{eqnarray}
 \Lambda_{\sigma}(\rho)= \Lambda_{\sigma}(\Delta(\rho))  =\sum_{\vec a}p_{\vec a}\Lambda_{\sigma}(\proj{\varphi_{\vec a}})
=\sum_{\vec a}p_{\vec a}\proj{\varphi_{t\vec a}}_A.
\end{eqnarray}
Similarly, for the purification $\Psi_{RA}$, we have 
\begin{eqnarray}
    I\ot\Lambda_{\sigma}(\proj{\Psi_{RA}})
    =\sum_{\vec a}p_{\vec a}\tau^R_{\vec a}\ot \proj{\varphi_{t\vec a}}_A,
\end{eqnarray}
where $\tau^R_{\vec a}=\Tr{\proj{\Psi_{RA}}I_R\ot\proj{\varphi_{\vec a}}_A}/p_{\vec a}$.
Hence, we have
\begin{eqnarray}
    S(\Lambda_{\sigma}(\rho))&=&S(\vec p),\\
    S( I\ot\Lambda_{\sigma}(\proj{\Psi_{RA}}))&=&S\left(\sum_{\vec a}p_{\vec a}\tau^R_{\vec a}\ot \proj{\varphi_{t\vec a}}_A\right)=S(\vec p)+\sum_{\vec a}p_{\vec a}S(\tau^R_{\vec a})\;.
\end{eqnarray}
Here $\vec p=\set{p_{\vec a}}_{\vec a\in \mathbb{Z}^n_d }$ is the probability distribution. Therefore, we have 
\begin{eqnarray}
I_c(\rho,\Lambda_{\sigma})\leq 0,
\end{eqnarray}
for any input state $\rho$.
That is, $Q^{(1)}(\Lambda_{\sigma})=0$.
Since $\Lambda^{\ot N}_{s,\sigma}=\Lambda_{s,\sigma^{\ot N}}$, and the tensor product of stabilizer states  is still a stabilizer
 state, we can repeat the above proof for $Q^{(1)}(\Lambda^{\ot N}_{\sigma})$ for any integer $N$. Hence, we have 
 $Q(\Lambda_{\sigma})=0$.

\end{proof}

\section{Discrete beam splitter with environment state being a magic state}\label{appen:dbs_mag}
\begin{thm}[Restatement of Theorem 3]
Given  nontrivial $s,t\in \mathbb{Z}_d$ with $s^2+t^2\equiv 1\mod d$, there exists 
some 1-qudit magic state $\ket{magic}$ a universal constant $c>0$, independent of $d$ and $n$, 
such that   
\begin{eqnarray}
Q(\Lambda_{s,\sigma_k})\geq kc,
\end{eqnarray}
where $\sigma_k$ is  a quantum state generated by any Clifford circuit $U_{cl}$ on 
$\ket{\text{magic}}^{\ot k}\ot \ket{0}^{n-k}$, i.e., $\sigma_k=U_{\text{cl}}\proj{\text{magic}}^{\ot k}\ot\proj{0}^{n-k}U^\dag_{\text{cl}}$.
\end{thm}
\begin{proof}

Since the Clifford unitary commutes  with the discrete beam splitter, see Lemma~\ref{lem:key_tech}.5,
then 
\begin{eqnarray}
 Q(\Lambda_{s,\sigma_k})=
  Q(\Lambda_{s,\proj{\text{magic}}^{\ot k}\ot\proj{0}^{n-k}})
  =Q(\ot^k\Lambda_{s,\ket{\text{magic}}}\ot^{n-k}\Lambda_{s,\ket{0}})
  \geq kQ(\Lambda_{s,\ket{\text{magic}}})\;.
\end{eqnarray}
Hence, we only need to consider the single-qudit case.
    Let us first consider the beam splitter  $U_{s,t}$ with $s^2\neq t^2 \mod d$, and single-qudit environment state $\sigma$ to be 
\begin{eqnarray}
    \ket{\sigma}_B=\frac{1}{\sqrt{2}}(\ket{0}_B+\ket{1}_B)\;,
\end{eqnarray}
and 
\begin{eqnarray}
    \ket{\psi}_{RA}=\frac{1}{\sqrt{2}}\left(\ket{0}_R\ket{0}_A+\ket{t^{-1}s}_R\ket{t^{-1}s}_A\right)\;.
\end{eqnarray}
Then we have 
\begin{eqnarray}
\ket{\tau}_{RAB}=U^{AB}_{s,t} \ket{\psi}_{RA}\ot\ket{\sigma}_B
=\frac{1}{2}
(
\ket{0,0,0}_{RAB}+\ket{0,t,s}_{RAB}+\ket{t^{-1}s,t^{-1}s^2,-s}_{RAB}+\ket{t^{-1}s,t^{-1},0}_{RAB}
)\;.
\end{eqnarray}
And thus
\begin{eqnarray}
\tau_{B}=\frac{1}{2}\proj{0}_B+\frac{1}{4}\proj{-s}_B+\frac{1}{4}\proj{s}_B,
\end{eqnarray}
and 
\begin{eqnarray}
    \tau_{A}=\frac{1}{4}\left(\proj{0}_A+\proj{t}_A+\proj{t^{-1}s^2}_A+\proj{t^{-1}}_A\right)\;.
\end{eqnarray}
Hence
\begin{eqnarray}
    I_c(\rho,\Lambda_{s,\sigma})
    =S(\tau_{A})-S(\tau_{RA})=S(\tau_A)-S(\tau_B)=\frac{1}{2}\;.
\end{eqnarray}

Now, let us consider the case where $s^2\equiv t^2\mod d$. This can be split into two cases: 
(a) $s\equiv t\mod d$ and (b) $s\equiv-t\mod d$.

Case (a). If $s\equiv t\mod d$, 
let us consider the environment state
\begin{eqnarray}
    \ket{\sigma}_B=\sqrt{\frac{2}{5}}\ket{0}_B+\sqrt{\frac{3}{5}}\ket{1}_B\;,
\end{eqnarray}
and 
\begin{eqnarray}
    \ket{\psi}_{RA}=\frac{\sqrt{6}}{5}\ket{0,0}_{RA}+\frac{3}{5}\ket{0,1}_{RA}+\sqrt{\frac{2}{5}}\ket{1,0}_{RA}\;.
\end{eqnarray}
Then we have 
\begin{eqnarray}
\ket{\tau}_{RAB}&=&U^{AB}_{s,s} \ket{\psi}_{RA}\ot\ket{\sigma}_B\\
&=&\left(\frac{\sqrt{6}}{5}\ket{0}_R+\sqrt{\frac{2}{5}}\ket{1}_R\right)\ot\left(\sqrt{\frac{2}{5}}\ket{0,0}_{AB}+\sqrt{\frac{3}{5}}\ket{s,s}_{AB}\right)
+\frac{3}{5}\ket{0}_R\ot\left(\sqrt{\frac{2}{5}}\ket{s,-s}_{AB}+\sqrt{\frac{3}{5}}\ket{2s,0}_{AB}\right)\;.
\end{eqnarray}
Thus 
\begin{eqnarray}
    \tau_{B}=\frac{59}{125}\proj{0}_B+\frac{6}{25}\proj{s}_B+\frac{36}{125}\proj{\phi}_B\;,
\end{eqnarray}
where $\ket{\phi}_B=\frac{1}{\sqrt{2}}(\ket{s}_B+\ket{-s}_B)$, with spectrum $\lambda_1=\frac{59}{125},\ \lambda_2=\frac{3(11-\sqrt{61})}{125},\ \lambda_3=\frac{3(11+\sqrt{61})}{125}$, and
\begin{eqnarray}
    \tau_A=\frac{4}{25}\proj{0}_A+\frac{66}{125}\proj{s}_A+\frac{39}{125}\proj{\varphi}_A\;,
\end{eqnarray}
where $\ket{\varphi}_A=\frac{2}{\sqrt{13}}\ket{0}_A+\frac{3}{\sqrt{13}}\ket{2s}_A$,
with spectrum $\mu_1=\frac{66}{125},\ \mu_2=\frac{59+\sqrt{1321}}{250},\mu_3=\frac{59-\sqrt{1321}}{250}$.
Hence, 
\begin{eqnarray}
    I_c(\rho,\Lambda_{s,\sigma})
    =S(\tau_{A})-S(\tau_{RA})=S(\tau_A)-S(\tau_B)\approx0.0178.
\end{eqnarray}

Case (b). If $s\equiv -t\mod d$, 
 consider the environment state
\begin{eqnarray}
    \ket{\sigma}_B=\sqrt{\frac{2}{5}}\ket{0}_B+\sqrt{\frac{3}{5}}\ket{-1}_B\;,
\end{eqnarray}
and 
\begin{eqnarray}
    \ket{\psi}_{RA}=\frac{\sqrt{6}}{5}\ket{0,0}_{RA}+\frac{3}{5}\ket{0,1}_{RA}+\sqrt{\frac{2}{5}}\ket{1,0}_{RA}\;.
\end{eqnarray}
Then we have 
\begin{eqnarray}
\ket{\tau}_{RAB}&=&U^{AB}_{s,-s} \ket{\psi}_{RA}\ot\ket{\sigma}_B\\
&=&\left(\frac{\sqrt{6}}{5}\ket{0}_R+\sqrt{\frac{2}{5}}\ket{1}_R\right)\ot\left(\sqrt{\frac{2}{5}}\ket{0,0}_{AB}+\sqrt{\frac{3}{5}}\ket{s,-s}_{AB}\right)
+\frac{3}{5}\ket{0}_R\ot\left(\sqrt{\frac{2}{5}}\ket{s,s}_{AB}+\sqrt{\frac{3}{5}}\ket{2s,0}_{AB}\right)\;.
\end{eqnarray}
Then the proof is the same as the Case (a).
\end{proof}

\begin{thm}[Restatement of Theorem 4]
Given nontrivial $s,t\in \mathbb{Z}_d$ with $s^2+t^2\equiv 1\mod d$, and the fixed environment state $\sigma$, the quantum capacity $Q(\Lambda_{s,\sigma})$ is bounded by the amount of magic resource in $\sigma$ as follows,
\begin{align}
    Q(\Lambda_{\sigma})\leq MRM(\sigma),
\end{align}
where $MRM(\sigma)$ is the modified relative entropy of magic.
\end{thm}

\begin{proof}

Since the modified relative entropy of magic $MRM(\sigma)=S(\mathcal{M}(\rho))-S(\rho)$ (See Theorem 4 in \cite{BGJ23a} or Lemma 37 in \cite{BGJ23c}), 
we only need to  prove the following result 
\begin{align}\label{ineq:key1}
\max_{\rho}S(\Lambda_{s,\sigma}(\rho))
    -S(\Lambda^c_{s,\sigma}(\rho))
    \leq MRM(\sigma)=S(\mathcal{M}(\sigma))-S(\sigma).
\end{align}
By using the equivalent definition in \eqref{eq:conv_B} and \eqref{eq:conv_C}, the above statement is equivalent to proving 
\begin{align}
    \max_{\rho}S( \rho \boxtimes\sigma )
    -S( \rho \tilde{\boxtimes} \sigma)\leq S(\mathcal{M}(\sigma))-S(\sigma).
\end{align}
Let us consider the stabilizer group $G_{\rho}$ of the mean state $\mathcal{M}(\sigma)$ with $r$ generator, which is equivalent to 
$\set{Z_1,...,Z_r}$ up to a Clifford unitary $U$. Hence, $\sigma=U(\proj{0}^{\ot r}\ot \sigma_E)U^\dag$, where
$E=\set{r+1, r+2, ..., n}$, Then 
\begin{align}
    \mathcal{M}(\sigma)
    =U\mathcal{M}(\proj{0}^{\ot r}\ot \sigma_E)U^\dag
    =U(\proj{0}^{\ot r}\ot \frac{I_E}{d^{n-r}})U^\dag,
\end{align}
and $S(\sigma)=S(\sigma_E)$.
Hence, $S(\mathcal{M}(\sigma))=\log d^{n-r}$.

Based on the invariance of quantum entropy under unitary, we have the following two equalities,  
\begin{align}
    S(\rho\boxtimes\sigma)=S(\rho\boxtimes (U\proj{0}^{\ot r}\ot \sigma_EU^\dag))
    =&S((U^\dag \rho U)\boxtimes (\proj{0}^{\ot r}\ot \sigma_E))=S(\rho'\boxtimes (\proj{0}^{\ot r}\ot \sigma_E)),\\
       S(\rho\tilde{\boxtimes}\sigma)=S(\rho\tilde{\boxtimes} (U\proj{0}^{\ot r}\ot \sigma_EU^\dag))
    =&S(U^\dag \rho U\tilde{\boxtimes} (\proj{0}^{\ot r}\ot \sigma_E))=S(\rho'\tilde{\boxtimes} (\proj{0}^{\ot r}\ot \sigma_E)),
\end{align}
where $\rho'=U^\dag \rho U$, $U$ is a Clifford unitary, and the second equality comes from the commutativity with Clifford unitaries in Lemma~\ref{lem:key_tech} (See  Lemma 85  in~\cite{BGJ23b}). 

In addition,  we have
\begin{align}
    \rho'\boxtimes (\proj{0}^{\ot r}\ot \sigma_E)
    =\sum_{\vec x}p_{\vec x} (\proj{\vec x}\boxtimes\proj{\vec 0}^{\ot r})\ot (\rho^E_{\vec x}\boxtimes\sigma_E),
\end{align}
where $p_{\vec x} =\Tr{\proj{\vec x}\ot I_E\rho'}$, and 
$\rho^E_{\vec x}=\bra{\vec x}\rho'\ket{\vec x}/p_{\vec x}$.
Similarly, 
\begin{align}
    \rho'\tilde{\boxtimes} (\proj{0}^{\ot r}\ot \sigma_E)
    =\sum_{\vec x}p_{\vec x} (\proj{\vec x}\tilde{\boxtimes}\proj{\vec 0}^{\ot r})\ot (\rho^E_{\vec x}\tilde{\boxtimes}\sigma_E).
\end{align}
Therefore, 
\begin{align}
  S(\rho\boxtimes\sigma)
-S(\rho\tilde{\boxtimes}\sigma)=&S(\rho'\boxtimes (\proj{0}^{\ot r}\ot \sigma_E))
-S(  \rho'\tilde{\boxtimes} (\proj{0}^{\ot r}\ot \sigma_E))\\
=&S(\vec p)+\sum_{\vec x}p_{\vec x}S(\rho^E_{\vec x}\boxtimes\sigma_E)
-(S(\vec p)+\sum_{\vec x}p_{\vec x}S(\rho^E_{\vec x}\tilde{\boxtimes}\sigma_E)\\
=&\sum_{\vec x}p_{\vec x}S(\rho^E_{\vec x}\boxtimes\sigma_E)
-\sum_{\vec x}p_{\vec x}S(\rho^E_{\vec x}\tilde{\boxtimes}\sigma_E)\\
\leq& \log d^{|E|}-S(\sigma_E)
=S(\mathcal{M}(\sigma))-S(\sigma),
\end{align}
where the inequality comes from the following inequalities
\begin{align}
    \label{ineq:ap1}S(\rho^E_{\vec x}\boxtimes\sigma_E)\leq& \log d^{|E|},\\
   \label{ineq:ap2} S(\rho^E_{\vec x}\tilde{\boxtimes}\sigma_E)\geq& S(\sigma_E).
\end{align}
Here, \eqref{ineq:ap1} come from the fact $\rho^E_{\vec x}\boxtimes\sigma_E$ is supported on $E$, 
and \eqref{ineq:ap2} comes from 
the entropy inequality $S(\rho\tilde{\boxtimes}\sigma)\geq \max\set{S(\rho), S(\sigma)}$
(See Theorem 14 in~\cite{BGJ23a}).
\end{proof}

\section{Discrete beam splitter with $s^2\equiv t^2\mod d$}\label{appen:s=t}
Let us define the quantum operation 
$\mathcal{A}$ as 
\begin{eqnarray}
    \mathcal{A}(\sigma)=A\sigma A^\dag\;.
\end{eqnarray}
Then discrere phase space inverse symmetry can be rewritten as 
\begin{eqnarray}
    \mathcal{A}(\sigma)=\sigma\;.
\end{eqnarray}

\begin{figure}[h]
  \center{\includegraphics[width=15cm]  {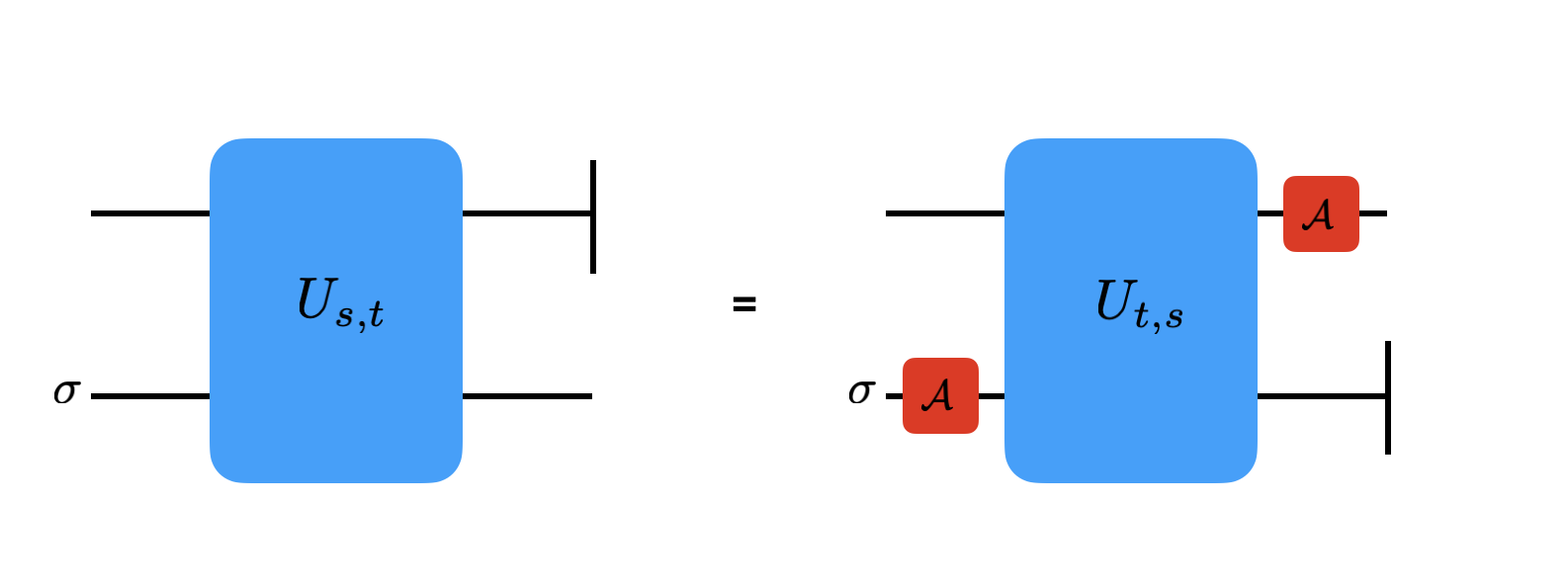}}     
  \caption{The diagram to show the equivalence of $\Lambda^c_{s,\sigma}$ and $
    \mathcal{A}\circ\Lambda_{t, \mathcal{A}(\sigma)}$.}
  \label{fig3}
 \end{figure}

\begin{lem}\label{lem:equ}
Given  $s^2+t^2\equiv 1 \mod d$, we have the relation
\begin{eqnarray}
    \Lambda^c_{s,\sigma}
    =\mathcal{A}\circ\Lambda_{t, \mathcal{A}(\sigma)}\;,
\end{eqnarray}
for any state $\sigma$ (See Fig. \ref{fig3}).
\end{lem}
\begin{proof}
Since the characteristic function $\Xi$ provides a complete description of the states and channels, we only need to  consider the characteristic function $  \Xi_{\Lambda^c_{s,\sigma}(\rho)}$ and $ \Xi_{\mathcal{A}\circ\Lambda_{t, \mathcal{A}(\sigma)}(\rho)}$ for any
 input state $\rho$.
Based on the  convolution-multiplication duality in~\cite{BGJ23a} stated here as Lemma~\ref{lem:key_tech}, we have 
    \begin{eqnarray}
    \Xi_{\Lambda^c_{s,\sigma}(\rho)}(\vec x)=\Xi_{\rho}(-t\vec x)\Xi_{\sigma}(s\vec x)\;,
    \end{eqnarray}
    and 
     \begin{eqnarray}
    \Xi_{\mathcal{A}\circ\Lambda_{t, \mathcal{A}(\sigma)}}(\vec x)=\Xi_{\rho}(-t\vec x)\Xi_{\mathcal{A}(\sigma)}(-s\vec x)
    =\Xi_{\rho}(-t\vec x)\Xi_{\sigma}(s\vec x)\;.
    \end{eqnarray}
Thus, we have $\Lambda^c_{s,\sigma}(\rho)
    =\mathcal{A}\circ\Lambda_{t, \mathcal{A}(\sigma)}(\rho)$ for any input state $\rho$, i.e., 
    $\Lambda^c_{s,\sigma}
    =\mathcal{A}\circ\Lambda_{t, \mathcal{A}(\sigma)}$.
\end{proof}

\begin{thm}[Restatement of Theorem~5]
Let $\sigma$ be an $n$-qudit state, which is some convex combination of $w(\vec a)\sigma_{\vec a} w(\vec a)^\dag$, with each state $\sigma_{\vec a}$ having discrete phase space inverse symmetry.
    Then $\Lambda_{s,\sigma}$ is anti-degradable for $s\equiv t\mod d$,  and
    \begin{eqnarray}
        Q(\Lambda_{s,\sigma})=0\;.
    \end{eqnarray}
\end{thm}

\begin{proof}
    We only need to consider the case where $\sigma=w(\vec a)\sigma_0w(\vec a)$ where $\sigma_0=\mathcal{A}(\sigma_0)$.
By Lemma \ref{lem:equ}, we have 
\begin{eqnarray}
    \Lambda^c_{s,\sigma}
    =\mathcal{A}\circ\Lambda_{t,\mathcal{A}(\sigma)}\;.
\end{eqnarray}
Moreover,
\begin{eqnarray}
    \mathcal{A}(\sigma)
    =\mathcal{A}(w(\vec a))\sigma_0\mathcal{A}(w(\vec a)^\dag)
    =w(-2\vec a)\sigma w(2\vec a)\;.
\end{eqnarray}
Let us denote $D_{\vec a}(\sigma)=w(\vec a)\sigma w(\vec a)^\dag$.
Hence $\mathcal{A}(D_{\vec a}(\sigma))=D_{-2\vec a}(\sigma)$. 
Besides, we have  
\begin{eqnarray}
    \rho\boxtimes_{t,s}D_{\vec a}(\sigma)=D_{s\vec a}(\rho\boxtimes_{t,s}\sigma)\;,
\end{eqnarray}
based on the following property(See Proposition 41 in~\cite{BGJ23b})
\begin{align*}
\Lambda(w(\vec x)\ot w(\vec y)\rho_{AB}w(\vec x)^\dag\ot w(\vec y)^\dag)
=w(s\vec x+t\vec y)\rho_{AB}w(s\vec x+t\vec y)^\dag\;,
\end{align*}
where $\Lambda(\rho_{I,II})=\Ptr{II}{U_{s,t}\rho_{I, II}U^\dag_{s,t}}$. 
Since $\rho\boxtimes_{t,s}D_{\vec a}(\sigma)=D_{s\vec a}(\rho\boxtimes_{t,s}\sigma)$, then 
\begin{eqnarray}
    \Lambda_{t, D_{\vec a}(\sigma)}=
    D_{s\vec a}\circ \Lambda_{t,\sigma}\;,
\end{eqnarray}
which implies that 
\begin{eqnarray}
    \Lambda^c_{s,\sigma}
    =\mathcal{T}\circ D_{-2s\vec a}\circ\Lambda_{t,\sigma}.
\end{eqnarray}
If $s\equiv t\mod d$, then $\Lambda_{t,\sigma}=\Lambda_{s,\sigma}$, and thus 
$\Lambda_{s,\sigma}$ is degradable.
\end{proof}

\section{Entanglement fidelity of the discrete beam splitter}

 Let us consider the encoding  $\mathcal{E}_K:\mathcal{H}_S\to \mathcal{H}_A$ and the decoding $\mathcal{D}_K:\mathcal{H}_A\to \mathcal{H}_S$, where $\mathcal{H}_S=\complex^K$ and  $\mathcal{H}_A=(\complex^d)^{\ot n}$.   
Entanglement fidelity is defined as 
\begin{align}
    F_e(\mathcal{E}_K, \mathcal{D}_K, \sigma)
    =\bra{\Phi}\mathcal{D}_K\circ\Lambda_{\sigma}\circ\mathcal{E}_K(\proj{\Phi})\ket{\Phi},
\end{align}
 where $\ket{\Phi}=\frac{1}{K}\sum_{i\in K}\ket{i}_S\ket{i}_R$ being maximally entangled state on $\mathcal{H}_S\ot\mathcal{H}_R$.
We explore how entanglement fidelity varies with different environment state $\sigma$.

\begin{prop}\label{prop:main1}
 The maximal entanglement fidelity over all stabilizer states$\sigma$, encoding 
 $\mathcal{E}_K$ and decoding $\mathcal{D}_K$ is
\begin{align}
\max_{\sigma\in STAB}
\max_{\mathcal{E}_K} \max_{\mathcal{D}_K}   F_e(\mathcal{E}_K, \mathcal{D}_K, \sigma)
=\frac{1}{K}.
\end{align}
\end{prop}
\begin{proof}
The proof is similar to that of Theorem \ref{thm:1_APP}.
First, since $\Lambda_{\lambda\sigma_1+(1-\lambda)\sigma_2}(\rho)=\lambda\Lambda_{\sigma_1}(\rho)+(1-\lambda)\Lambda_{\sigma_2}$, 
then by the linearity, we get 
\begin{align}
    F_e(\mathcal{E}_K, \mathcal{D}_K, \lambda\sigma_1+(1-\lambda)\sigma_2)
    =  \lambda F_e(\mathcal{E}_K, \mathcal{D}_K,\sigma_1)
    + (1-\lambda)F_e(\mathcal{E}_K, \mathcal{D}_K, \sigma_2).
\end{align}
Hence, we only consider the case where $\sigma$ is pure 
stabilizer state.

For pure stabilizer state $\sigma$, 
 there exists some abelian group of Weyl operators with $n$ generators $G_{\sigma}:=\set{w(\vec x): w(\vec x)\sigma=\sigma}$.
By a simple calculation, $\Delta(\rho):=\frac{1}{d^n}\sum_{w(\vec x)\in G_{\sigma}}w(\vec x)\rho w(\vec x)^\dag$ is the full-dephasing channel, and there exists 
an orthonormal basis $\ket{\varphi_{\vec a}}$ such that 
$$w(\vec x)\ket{\varphi_{\vec a}}=\omega^{\vec x\cdot\vec a}_d\ket{\varphi_{\vec a}}, $$ for any $w(\vec x)\in G_{\sigma}$, and 
$$\Delta(\rho)=\frac{1}{d^n}\sum_{\vec a\in \mathbb{Z}^n_d}\bra{\varphi_{\vec a}}\rho\ket{\varphi_{\vec a}}\proj{\varphi_{\vec a}}.$$ Note that, $\sigma=\proj{\varphi_{\vec 0}}$. Let us denote $p_{\vec a}=\bra{\varphi_{\vec a}}\rho\ket{\varphi_{\vec a}}$.
Then, we have 
\begin{eqnarray}
\Lambda_{\sigma}(\rho)=\Lambda(\rho\ot\sigma)
=\frac{1}{d^n}\sum_{w(\vec x)\in G_{\sigma}}\Lambda(\rho\ot w(\vec x)\sigma w(\vec x)^\dag)
=\Lambda\left(\frac{1}{d^n}\sum_{w(\vec x)\in G_{\sigma}}w(\vec x)\rho w(\vec x)^\dag\ot \sigma \right)
=\Lambda_{\sigma}(\Delta(\rho)),
\end{eqnarray}
where the second equality comes from the definition of $G_{\sigma}$, and 
the 
third equality comes from the following property (proved as Proposition 41 in\cite{BGJ23b}), namely 
\begin{align*}
\Lambda(w(\vec x)\ot w(\vec y)\rho_{AB}w(\vec x)^\dag\ot w(\vec y)^\dag)
=w(s\vec x+t\vec y)\rho_{AB}w(s\vec x+t\vec y)^\dag.
\end{align*}

Hence, 
\begin{align}
    \Lambda_{\sigma}\circ \mathcal{E}(\proj{\Phi}_{SR})
    =\Lambda_{\sigma}(\Delta_A (\proj{\rho_{\mathcal{E}}}_{SR}))
    =\sum_{\vec a}p_{\vec a}\Lambda_{\sigma}(\proj{\varphi_{\vec a}})_A\ot\tau^R_{\vec a}.
\end{align}
where $\rho^{\mathcal{E}}_{AR}=\mathcal{E}(\proj{\Phi}_{SR})$ is the encoded state of $\ket{\Phi}_{SR}$, 
$p_{\vec a}=\Tr{\rho^{\mathcal{E}}_{AR}\proj{\varphi_{\vec a}}}$,
and $\tau^R_{\vec a}=\Ptr{A}{\rho^{\mathcal{E}}_{AR}\proj{\varphi_{\vec a}}}/p_{\vec a}$.
Thus,
\begin{align}
    F_e(\mathcal{E}_K, \mathcal{D}_K, \sigma)
    =\bra{\Phi}\mathcal{D}_K\circ\Lambda_{\sigma}\circ\mathcal{E}_K(\proj{\Phi})\ket{\Phi}
    =\sum_{\vec a}p_{\vec a}\Tr{\mathcal{D}\circ\Lambda_{\sigma}(\proj{\varphi_{\vec a}})_A\ot\tau^R_{\vec a}\proj{\Phi}_{SR}}
    \leq \frac{1}{K},
\end{align}
where the last inequality comes from the fact that
\begin{eqnarray}
    \Tr{\mathcal{D}\circ\Lambda_{\sigma}(\proj{\varphi_{\vec a}})_A\ot\tau^R_{\vec a}\proj{\Phi}_{SR}}
    \leq \max_{\rho_S, \sigma_R}\Tr{\rho_S\ot\sigma_R \proj{\Phi}_{SR}}
    \leq \frac{1}{K},
\end{eqnarray}
as the largest Schmidt coefficient of the maximal entangled state $\ket{\Phi}_{SE}$ is $1/\sqrt{K}$. Hence, for the pure stabilizer state $\sigma$,
\begin{align}
    \max_{\mathcal{E}_K} \max_{\mathcal{D}_K}   F_e(\mathcal{E}_K, \mathcal{D}_K, \sigma)
\leq \frac{1}{K}.
\end{align}

Now, let construct an example with stabilizer environment state $\sigma$, encoding 
$\mathcal{E}_K$ and encoding $\mathcal{D}_K$, such that the entanglement fidelity is 
$\frac{1}{K}$.
Let us consider the $n$-qudit stabilizer state $\sigma=\proj{\vec 0}$, and consider the 
encoding
\begin{align}
    \mathcal{E}: \ket{i}_S\to \ket{\vec x_i}_A,
\end{align}
where $\vec x_i\in\mathbb{Z}^n_d$, and $\vec x_i\neq \vec x_j$ for $i\neq j$.
Hence, we have 
\begin{align}
    \Lambda_{\sigma}(\ket{\vec x_i}\bra{\vec x_j}_A)
    =\Ptr{B}{U_{s,t}\ket{\vec x_i}\bra{\vec x_j}_A\ot \proj{\vec 0}_BU^\dag_{s,t}}
    =\Ptr{B}{\ket{s\vec x_i}\bra{s\vec x_j}_A\ot \ket{-t\vec x_i}\bra{-t\vec x_j}_B}
    =0, \forall i \neq j,
\end{align}
and 
\begin{align}
    \Lambda_{\sigma}(\ket{\vec x_i}\bra{\vec x_i}_A)
    =\proj{s\vec x_i}_A, \forall i\in [K].
\end{align}
Therefore, 
\begin{align}
    \Lambda_{\sigma}\circ\mathcal{E} (\proj{\Phi}_{AR})
    =\frac{1}{K}\sum_{i,j}   \Lambda_{\sigma}(\ket{\vec x_i}\bra{\vec x_j})\ot \ket{i}\bra{j}_R
=\frac{1}{K}\sum_i \Lambda_{\sigma}(\ket{\vec x_i}\bra{\vec x_i})\ot \ket{i}\bra{i}_R.
\end{align}
% which comes from the fact that for any $i\neq j$,
% \begin{align}
%     \Lambda_{\sigma}(\ket{\vec x_i}\bra{\vec x_j})
%     =\Ptr{B}{U_{s,t}\ket{\vec x_i}\bra{\vec x_j}_A\ot \proj{\vec 0}_BU^\dag_{s,t}}
%     =\Ptr{B}{\ket{s\vec x_i}\bra{s\vec x_j}_A\ot \ket{-t\vec x_i}\bra{-t\vec x_j}_B}
%     =0.
% \end{align}
% In addition,
% \begin{align}
%     \Lambda_{\sigma}(\ket{\vec x_i}\bra{\vec x_i})
%     =\proj{s\vec x_i}_A, \forall i,
% \end{align}
Let us take the decoding $\mathcal{D}_K$
as follows
\begin{align}
   \mathcal{D}_K: \ket{s\vec x_i}_A\to \ket{i}_S.
\end{align}
Then, the state after the decoding $\mathcal{D_K}$ will become
\begin{align}
    \mathcal{D} \circ\Lambda_{\sigma}\circ\mathcal{E} (\proj{\Phi}_{AR})
  =&  \frac{1}{K}\sum_i  \mathcal{D} \circ\Lambda_{\sigma}(\ket{\vec x_i}\bra{\vec x_i})\ot \proj{i}_R\\
  =&\frac{1}{K}\sum_i  \mathcal{D}\circ \Lambda_{\sigma}(\proj{\vec x_i})\ot \proj{i}_R\\
  =&\frac{1}{K}\sum_i  \mathcal{D} (\proj{s\vec x_i})\ot \proj{i}_R\\
  =&\frac{1}{K}\sum_i  \proj{i}_S\ot \proj{i}_R.
\end{align}
Therefore
\begin{align}
    F_e(\mathcal{E}_K, \mathcal{D}_K, \sigma)
    =\bra{\Phi}\mathcal{D}_K\circ\Lambda_{\sigma}\circ\mathcal{E}_K(\proj{\Phi})\ket{\Phi}
    =\frac{1}{K}\sum_i  \bra{\Phi}\proj{i}\ot \proj{i}_R\ket{\Phi}
    =\frac{1}{K}.
\end{align}
\end{proof}

For $K=2$, i.e., $\mathcal{H}_S$ is a qubit, 
 let us consider the magic environment state 
\begin{align}\label{exam_1}
   \ket{\sigma}_B=\frac{1}{\sqrt{2}}(\ket{0}_B+\ket{t}_B),
\end{align} 
and the stabilizer encoding
$\mathcal{E}_2$ as follows
\begin{align}\label{def:enc}
    \mathcal{E}_2:  \ket{0}_S\to \ket{0}_A, \quad
    \ket{1}_S\to \ket{s}_A.
\end{align}

\begin{prop}\label{prop:main2}
Given nontrivial $s,t$ with $s^2\not\equiv t^2\mod d$, 
a magic environment state $\sigma$ in \eqref{exam_1},  encoding $\mathcal{E}_2:\mathcal{C}^2\to\mathcal{C}^d$ in \eqref{def:enc}, there exists a  decoding 
$\mathcal{D}:\mathcal{C}^d\to \mathcal{C}^2$ such that 
    \begin{align}
      F_e(\mathcal{E}_2, \mathcal{D}_2, \sigma)
= \frac{3}{4}.
    \end{align}
\end{prop}
\begin{proof}
Let us consider the 
maximal entangled state 
\begin{align}
    \ket{\Phi}_{SR}
=\frac{1}{\sqrt{2}}(\ket{0}_S\ket{0}_R+\ket{1}_S\ket{1}_R).
\end{align}
After the encoding in \eqref{def:enc}, the maximal entangled state  will become
    \begin{eqnarray}
\ket{\Phi_{\mathcal{E}}}_{AR}
=\frac{1}{\sqrt{2}}(\ket{0}_A\ket{0}_R+\ket{s}_A\ket{1}_R),
\end{eqnarray}
where before encoding $\ket{\Phi}_{SR}=\frac{1}{\sqrt{2}}(\ket{0}_S\ket{0}_R+\ket{1}_S\ket{1}_R)$.
Thus 
\begin{align}
\ket{\tau}_{ABR}
=&U^{AB}_{s,t}\ket{\Psi_{\mathcal{E}}}_{AR}
\ket{\phi}_B
=\frac{1}{2}
\sum^{2}_{k=1}
\ket{kt^2}_A\ket{kst}_B\ket{0}_E
+\frac{1}{2}
\sum^{2}_{k=1}
\ket{s^2+kt^2}_A
\ket{(k-1)st}_B
\ket{1}_R\\
=&\frac{1}{2}
\sum^{2}_{k=1}
\ket{kt^2}_A\ket{kst}_B\ket{0}_R
+\frac{1}{2}
\sum^{2}_{k=1}
\ket{(k-1)t^2+1}_A
\ket{(k-1)st}_B
\ket{1}_R.
\end{align} 
Hence, 
\begin{eqnarray}
    \Lambda_{\sigma}\circ \mathcal{E}_2(\proj{\Phi}_{SR})
    =\frac{1}{2}\proj{\mu}_{AR}+\frac{1}{4}\proj{2t^2}_A\ot\proj{0}_E
    +\frac{1}{4}\proj{0}_A\ot\proj{1}_R,
\end{eqnarray}
where 
\begin{align}
    \ket{\mu}_{AR}
    =\frac{1}{\sqrt{2}}(\ket{t^2}_A\ket{0}_R+\ket{t^2+1}_A\ket{1}_R).
\end{align}
Since $t^2\neq 0,1\mod d$, the quantum states $\ket{0}_A,\ket{t^2},\ket{2t^2}, \ket{t^2+1}$ are orthogonal to each other.
Hence, we consider the decoding $\mathcal{D}_2$ 
which will map $\ket{kt^2}_A$ to $\ket{0}_S$ for all $k\in \set{0,1,2}$, and 
 $\ket{t^2+1}_A$ to $\ket{1}_S$. 
After the decoding $\mathcal{D}_2$, the state will become 
\begin{align}
 \mathcal{D}_2\circ\Lambda_{\sigma}\circ \mathcal{E}_2(\proj{\Phi}_{SR})
 =\frac{1}{2}\proj{\Phi}_{SR}
 +\frac{1}{4}\proj{0}_S\ot\proj{0}_R
    +\frac{1}{4}\proj{1}_S\ot\proj{1}_R,
\end{align}
and thus 
\begin{align}
  F_e(\mathcal{E}_2, \mathcal{D}_2, \sigma)
  =\bra{\Phi}\mathcal{D}_2\circ\Lambda_{\sigma}\circ \mathcal{E}_2(\proj{\Phi}_{SR})\ket{\Phi}
=\frac{1}{2}
+\frac{1}{4}
=\frac{3}{4}.
\end{align}

\end{proof}
Hence, for $K=2$, by the Propositions \ref{prop:main1} and 
\ref{prop:main2}, we find that the entanglement fidelity 
with magic environmental state \eqref{exam_1} is larger than that with any stabilizer states,
\begin{align*}
    F_e(\mathcal{E}_2, \mathcal{D}_2, \sigma_B)=\frac{3}{4}>\frac{1}{2}=\max_{\tau\in STAB}\max_{\mathcal{E}_2} \max_{\mathcal{D}_2}   F_e(\mathcal{E}_2, \mathcal{D}_2, \tau).
\end{align*}

\begin{prop}
    The advantage of  magic state on the performance of 
entanglement fidelity compared to stabilizer states  is bounded by the magic amount as follows
\begin{align}
    \frac{\max_{\mathcal{E}_K} \max_{\mathcal{D}_K}   F_e(\mathcal{E}_K, \mathcal{D}_K, \sigma)}{\max_{\tau\in STAB}\max_{\mathcal{E}_K} \max_{\mathcal{D}_K}   F_e(\mathcal{E}_K, \mathcal{D}_K, \tau)}
    \leq 2^{MRM_{\infty}(\sigma)}.
\end{align}
\end{prop}
\begin{proof}
Based on  definition of magic measure 
$MRM_{\infty}(\sigma)=\min_{\sigma\in STAB}D_{\infty}(\sigma||\tau)$, there exists some 
stabilizer state $\tau$ such that 
\begin{align}
    \sigma \leq 2^{MRM_{\infty}(\sigma)}\tau.
\end{align}
Then, for encoding $\mathcal{E}_K$ and decoding $\mathcal{D}_K$, 
\begin{align}
    F_e(\mathcal{E}_K, \mathcal{D}_K, \sigma)
    \leq 2^{MRM_{\infty}(\sigma)} F_e(\mathcal{E}_K, \mathcal{D}_K, \tau)
    \leq 2^{MRM_{\infty}(\sigma)} \max_{\tau\in STAB}F_e(\mathcal{E}_K, \mathcal{D}_K, \tau).
\end{align}
Hence, we get the result by taking the maximum over all encoding $\mathcal{E}_K$ and decoding $\mathcal{D}_K$ on both sides
of the above inequality.

\end{proof}

\end{appendix}
\end{document}